\documentclass[preprint,12pt,a4paper,twoside,fleqn,sort&compress]{elsarticle}
\usepackage{amssymb,latexsym,amscd,amsmath,amsthm}
\usepackage{pst-all}
\usepackage[varg]{pxfonts}
\usepackage{mathrsfs}
\usepackage{shuffle}
\usepackage{hyperref}
\usepackage{times}

\usepackage{color} 

\usepackage{mathrsfs}
\usepackage{graphicx}

\usepackage{helvet}
\usepackage{courier}

\usepackage{algorithm}
\usepackage[noend]{algorithmic}
\usepackage{verbatim}
\usepackage{flafter}
\newcommand{\rw}{\rightarrow}

\theoremstyle{definition}
\newtheorem{definition}{Definition}[section]
\theoremstyle{plain}
\newtheorem{proposition}{Proposition}[section]
\newtheorem{theorem}{Theorem}[section]
\newtheorem*{thm-other}{Theorem}

\theoremstyle{remark}
\newtheorem{remark}{Remark}

\begin{document}

\journal{arXiv}

\begin{frontmatter}


\title{Computation Tree Logic Model Checking Based on Possibility Measures
\thanks{This work was partially supported by National Science
Foundation of China (Grant No: 11271237,61228305) and the Higher School Doctoral
Subject Foundation of Ministry of Education of China (Grant No:20130202110001).}}

\author {Yongming Li\corref{cor1}}
\ead{liyongm@snnu.edu.cn}

\author {Yali Li}
\author {Zhanyou Ma}

\address {College of Computer Science, Shaanxi Normal University, Xi'an, 710062, China}
\cortext[cor1]{Corresponding Author}
\begin{abstract}
In order to deal with the systematic verification with uncertain infromation in possibility theory, Li and Li \cite{li12} introduced model checking of linear-time properties in which the uncertainty is modeled by possibility measures. Xue, Lei and Li \cite{Xue09} defined computation tree logic (CTL) based on possibility measures, which is called possibilistic CTL (PoCTL). This paper is a continuation of the above work. First, we study the expressiveness of PoCTL. Unlike probabilistic CTL, it is shown that PoCTL (in particular, qualitative PoCTL) is more powerful than CTL with respect to their expressiveness. The equivalent expressions of basic CTL formulae using qualitative PoCTL formulae are presented in detail. Some PoCTL formulae that can not be expressed by any CTL formulae are presented. In particular, some qualitative properties of repeated reachability and persistence are expressed using PoCTL formulae. Next, adapting CTL model-checking algorithm, a method to solve the PoCTL model-checking problem and its time complexity are discussed in detail. Finally, an example is given to illustrate the PoCTL model-checking method.
\end{abstract}


\end{frontmatter}

\baselineskip 20pt

\section{Introduction}

Model checking \cite{EGP92} is a formal
verification technique which allows for desired behavioral properties of a given system to
be verified on the basis of a suitable model of the system through systematic inspection
of all states of the model. It is widely used in the design and analysis of computer systems \cite{Dovier02,Clarke08}. Although it has been rapidly gaining in importance in recent years, classical model checking can not deal with verification of those systems (e.g.,concurrent systems) dealing with uncertainty information. Such as, the development of most large and complex systems is inevitably involved with lots of  uncertainty and inconsistency information.

In order to handle the systematic
verification with uncertain information in probability,  Hart and Sharir \cite{Hart06} in 1986
applied probability theory to model
checking  in which the uncertainty is modeled by probability measures. Baier and Katoen \cite{Baier08} systematically introduced the principle and method of model checking based on probability
measures and related applications with Markov chain models for probabilistic systems.
For the past few years, there were even more applications on probability model checking  in verifying properties of systems with uncertain information (see e.g. \cite{Barbuti00}).

On the other hand,
 Zadeh proposed the theory of fuzzy sets in 1965 \cite{Zadeh65}, and possibility measures \cite{Sugeno74,Zadeh78} are a development of classical measures as a branch of the theory of fuzzy
sets from then. As a comparison, possibility measures (more general, fuzzy measures) focus on non-additive situation, while probability measures are used for additive situation. Most problems in real situations are complicated and non-additive. As a matter of fact, fuzziness seems to pervade most human perception and thinking processes as noted by Zadeh, especially, modeling human-centered systems, including biomedical systems (\cite{lin02}), criminal trial systems, decision making systems(\cite{grabisch00}), linguistic quantifiers (\cite{ying06,cui08}), and knowledge base (\cite{dubois94}). Therefore, it is necessary to study the theory and its applications of model checking on non-deterministic systems of non-additive measure, especially, fuzzy measure. In this respect, Li and Li \cite{li12} introduced model checking of linear-time properties in which the uncertainty is modeled by possibility measures and initiated the model checking based on possibility measures. Xue, Lei and Li \cite{Xue09} defined computation tree logic based on possibility measures, which is called possibilistic computation tree logic (PoCTL, in short).

 Although  we have studied the quantitative and qualitative properties of PoCTL in \cite{Xue09}, there are many important issues that still have not been addressed. The first important problem is the expressiveness of PoCTL: whether any CTL formulae can be expressed by PoCTL or vise versa. As we know, probabilistic CTL and CTL are not comparable with each other (\cite{Baier08}). This allows probabilistic CTL to be used to do model checking of real-world problems, which can not be tackled by classical CTL model checking. The surprising result of this paper is that CTL is a proper subclass of PoCTL. The second problem is looking for the method to solve PoCTL model-checking problems. As we know, there are effective algorithms and automated tools to solve CTL model-checking problems. As we just mentioned, CTL is a proper subclass of PoCTL, it is nontrivial to study whether there are effective algorithms to solve the PoCTL model-checking problems. We shall give complete study to the above two problems in this paper.

The content of this paper is arranged as follows. In Section 2 we recall the notion of possibilistic Kripke structures, the related possibility measures induced by the possibilistic Kripke structures, and the main notions of PoCTL introduced in \cite{Xue09}. In Section 3, the equivalence of PoCTL formulae and CTL formulae is investigated, and the differences between PoCTL formulae and CTL  formulae are discussed. An important result, CTL is a proper subclass of PoCTL, is obtained.
 Section 3 also presents qualitative properties of repeated reachability and persistence.
The PoCTL model checking approach is presented in Section 4, and an illustrative example is given in Section 5.
The paper ends with conclusion section.

\section {Preliminaries}

Transition systems or Kripke structures are key models for model checking. Corresponding to possibilistic model checking, we have the notion of possibilistic Kripke structures, which is defined as follows.

\begin{definition} \cite{li12}\label{de:pkripke}
A possibilistic Kripke structure is a tuple $M=(S,P,I,AP,L)$, where

(1) $S$  is a countable, nonempty set of states;

(2) $P:S\times S\longrightarrow [0,1]$ is the transition possibility distribution such that for all states $s$, $\bigvee\limits_{s^{'}\in S}P(s,s^{'})=1$ ;

(3) $I:S\longrightarrow[0,1]$ is the initial distribution, such that $\bigvee\limits_{s\in S}I(s)=1$ ;

(4) $AP$ is a set of atomic propositions;

(5) $L:S\longrightarrow 2^{AP}$  is a labeling function that labels a state $s$ with those atomic propositions in $AP$ that are supposed to hold in $s$.

Furthermore, if the set  $S$  and  $AP$ are finite sets, then $M=(S,P,I,AP,L)$ is called a
finite possibilistic Kripke structure.
\end{definition}

\begin{remark} (1) In Definition \ref{de:pkripke}, we require the transition possibility distribution and initial distribution are normal, i.e., $\vee_{s'\in S}P(s,s')=1$ and $\vee_{s\in S}I(s)=1$, where we use $\vee X$ or $\wedge X$ to represent the least upper bound (or supremum) or the largest lower bound (or infimum) of the subset $X\subseteq [0,1]$, respectively. These conditions are corresponding to the transition probability distribution and probability initial distribution in probabilistic Kripke structure or Markov chain (\cite{Baier08}), where the supremum operation is replaced by the sum operation. They are the main differences between possibilistic Kripke structure and probabilistic Kripke structure. In fact, in fuzzy uncertainty, the order instead of the additivity is one of the most important factors to be considered.

(2) The transition possibility distribution $P: S\times S\longrightarrow [0,1]$ can also be represented by a fuzzy matrix. For convenience, this fuzzy matrix is also written as $P$, i.e., $$P=(P(s,t))_{s,t\in S},$$ and $P$ is also called the (fuzzy) transition matrix of $M$. In \cite{li12}, we also used the symbol $A$ to represent transition matrix. For the fuzzy matrix $P$, its transitive closure is denoted by $P^+$. When $S$ is finite, and if $S$ has $N$ elements, i.e., $N=|S|$, then $P^+=P\vee P^2\vee\cdots\vee P^N$ \cite{li05}, where $P^{k+1}=P^k\circ P$ for any positive integer number $k$. Here, we use the symbol $\circ$ to represent the max-min composition operation of fuzzy matrixes. Recall that the max-min composition operation
of fuzzy matrices is similar to ordinary matrix multiplication operation, that is, let ordinary multiplication and addition operations of real numbers be replaced by minimum and maximum operations of real numbers (\cite{Zadeh78}).

For a possibilistic Kripke structure $M=(S,P,I,AP,L)$, using $P^+$, we can get another possibilistic Kripke structure $M^+=(S,P^+,I,AP,L)$.

(3) The authors in \cite{hajek95} also used the notion of fuzzy possibilistic Kripke structures as the models of qualitative possibilistic logic QFL, which is formally defined as a structure $K=(W,\Vdash,\pi)$ where $W$ is a nonempty set of worlds, $\Vdash$ maps $AP\times W$ into the truth value set $\{0,1/n,2/n,\cdots,1\}(n\geq 1)$, and $\pi$ is a normalized positive fuzzy subset of $W$, i.e., a mapping $\pi: W\longrightarrow [0,1]$ such that $\pi(w)>0$ for each $w$ and $\bigvee_{w\in W}\pi(w)=1$. Obviously, the notion of fuzzy possibilistic Kripke structure just defined is not equivalent to our notion of possibilistic Kripke structures. Since our notion of possibilistic Kripke structures is obvious a generalization of classical Kripke structures (see \cite{EGP92}) into fuzzy cases and a possibilistic version of (discrete-time) Markov chains as defined in Definition 10.1 in \cite{Baier08}. So we still use the name of possibilistic Kripke structures here, but it has no connection with that defined in \cite{hajek95}. The much more related notion is (discrete-time) fuzzy Markov chains \cite{kruse87} or (discrete-time) possibilistic Markov chains (\cite{dubois94}) or possibilistic Markov processes (\cite{janssen96}) which are used to model certain fuzzy systems. The only difference between possibilistic Kripke structures and fuzzy (or possibilistic) Markov chains lies in that there is no labeling function in the definition of fuzzy (or possibilistic) Markov chains. In \cite{dubois94}, possibilistic Markov chains are used to model the evolution of updating problem in a knowledge base that describes the state of evolving system. Uncertainty comes from incomplete knowledge about the knowledge base, ``one may only have some idea about what is/are the most plausible state(s) of the system, among possible one''(\cite{dubois94}). This type of incomplete knowledge was described in terms of possibility distribution in \cite{dubois94}, the degree of transition possibility distribution denotes the plausible degree of the next state. This provides us a sort of justification for degrees of transitions in possibilistic Kripke structures.
\end{remark}

The states $s$ with $I(s)>0$ are considered as the initial states. For state $s$ and $T\subseteq S$, let $P(s,T)$ denote the possibility of moving from $s$ to some state $t\in T$ in a single step, that is,
\begin{equation*}
P(s,T)=\vee_{t\in T}P(s,t).
\end{equation*}

Paths in possibilistic Kripke structure $M$ are infinite paths in the underlying digraph. They are
defined as infinite state sequences $\pi=s_{0}s_{1}s_{2}\cdots\in S^{w}$  such that  $P(s_{i},s_{i+1})>0$ for all $i\in I$.
Let $Paths(M)$ denote the set of all paths in $M$, and $Paths_{fin}(M)$ denote the set of finite path
fragments $s_{0}s_{1}\cdots s_{n}$ where $n\geq 0$ and $P(s_{i},s_{i+1})>0$ for $0\leq i\leq n$. Let $Paths_M(s)$ ($Paths(s)$ if $M$ is understood) denote the set of all
paths in $M$ that start in state $s$. Similarly, $Paths_{M-fin}(s)$ ($Paths_{fin}(s)$ if $M$ is understood) denotes the set of finite path fragments $s_{0}s_{1}\cdots s_{n}$ such that $s_{0}=s$.
The set of direct successors (called  $Post$ ) and direct predecessors (named  $Pre$ ) are defined
as follows:
\begin{equation*}
Post(s)=\{s'\in S\mid P(s,s')> 0\};~~ Pre(s)=\{s' \in S\mid P(s',s)> 0\}.
\end{equation*}

Given a possibilistic Kripke structure $M$, the cylinder set of $\hat{\pi}=s_0\cdots s_n\in Paths_{fin}(M)$ is defined as (\cite{Baier08}) $$Cyl(\hat{\pi})=\{\pi\in Paths(M) | \hat{\pi}\in Pref(\pi)\},$$ where $Pref(\pi)=\{\pi^{\prime} | \pi^{\prime}$ is a finite prefix of $\pi\}$. Then as shown in \cite{li12},  $\Omega=2^{Paths(M)}$ is the algebra generated by $\{Cyl(\hat{\pi})\mid\hat{\pi}\in Paths_{fin}(M)\}$ on $Paths(M)$. That is to say, $\Omega=2^{Paths(M)}$ is the unique subalgebra of $2^{Paths(M)}$ which is closed under unions and intersections containing $\{Cyl(\hat{\pi}) | \hat{\pi}\in Pref(\pi)\}$.

\begin{definition}\label{def:possibility measure} \cite{li12} For a possibilistic Kripke structure $M$, a function $Po^M: Paths(M)\rightarrow [0,1]$ is defined as follows:
\begin{equation}\label{eq:possibility measure-path}
Po^M(\pi)=I(s_{0})\wedge\bigwedge\limits_{i=0}^\infty P(s_{i},s_{i+1})
\end{equation}
for any $\pi=s_{0}s_{1} \cdots, \pi\in Paths(M).$
Furthermore, we define
\begin{equation}\label{eq:possibility measure}
Po^M(E)=\vee\{Po^M(\pi)\mid\pi\in E\}
\end{equation}
for any $E\subseteq Paths(M)$, then, we have a well-defined function $$Po^M:2^{Paths(M)}\longrightarrow [0,1],$$ $Po^M$ is called the possibility measure over $\Omega=2^{Paths(M)}$ as it has the properties stated in Theorem \ref{th:possibility measure}. If $M$ is clear from the context, then $M$ is omitted and we simply write $Po$ for $Po^M$.

\end{definition}

\begin{theorem}\label{th:possibility measure} \cite{li12} $Po$ is a possibility measure on $\Omega=2^{Paths(M)}$, i.e., $Po$ satisfies the following conditions:

 (1) $Po(\varnothing)=0$, $Po(Paths(M))=1$;

 (2) $Po(\bigcup\limits_{i\in I}A_{i})=\bigvee\limits_{i\in I}Po(A_{i})$ for any $A_{i}\in\Omega$, $i\in I$.

\end{theorem}

\begin{theorem}\label{th:possibility on cyl} \cite{li12} Let $M$ be a  possibilistic finite Kripke structure. Then the possibility measure of the cylinder sets is given by $Po(Cyl(s_{0}\cdots s_{n}))=I(s_{0})\wedge \bigwedge\limits_{i=0}^{n-1} P(s_{i},s_{i+1})$ when $n>0$ and $Po(Cyl(s_{0}))=I(s_{0})$.
\end{theorem}

 \begin{remark}\label{re:m-s}
     (1) For paths starting in a certain (possibly noninitial) state $s$, the same construction is applied to the possibilistic Kripke structure
 $M_s$ that results from $M$ by letting $s$ be the unique initial state. Formally, for $M=(S,P,I,AP,L)$ and state $s$,  $M_s$ is defined by $M_s=(S,P,s,AP,L)$ , where $s$ denotes an initial distribution with only one initial state $s$.

 (2) For a probabilistic Kripke structure $M$, by the intension property of probability measures, the induced probability measure (\cite{Baier08}), which is defined on the $\sigma$-algebra of $2^{Paths(M)}$ generated by cylinder sets, is uniquely determined by its definition on cylinder sets. On the other hand, by the extensional property of possibility measures, the induced possibility measure in Eq. (\ref{eq:possibility measure}) is uniquely determined by its definition on single paths as shown in Eq.(\ref{eq:possibility measure-path}). The method to define  probability measure on a probabilistic Kripke structure can not be applied to that of possibility measure on possibilistic Kripke structure, and vice versa. For more comparisons of possibility measures and probability measures, we refer to \cite{Didier14,Drakopoulos11,grabisch00,li12} and references therein.
\end{remark}

\begin{definition} \cite{Xue09} (Syntax of PoCTL) {\sl PoCTL state formulae} over the set $AP$ of atomic propositions are formed according to the following grammar:
\begin{center}
$\Phi ::= true\mid a \mid\Phi_{1} \wedge \Phi_{2}\mid \neg \Phi\mid Po_{J}(\varphi)$
\end{center}
where $a\in AP$, $\varphi$ is a PoCTL path formula and $J\subseteq [0,1]$ is an interval with rational bounds.

{\sl PoCTL path formulae} are formed according to the following grammar:

\begin{center}
$\varphi::=\bigcirc \Phi \mid \Phi_{1}\sqcup \Phi_{2}\mid \Phi_{1}\sqcup^{\leq n} \Phi_{2} $
\end{center}
where $ \Phi$, $ \Phi_{1}$, and $ \Phi_{2}$ are state formulae and $n\in\mathbb{N}$.
\end{definition}

\begin{definition}\label{def:semantics of PoCTL}\cite{Xue09} (Semantics of PoCTL) Let $a\in AP$ be an atomic proposition, $M=(S,P,I,AP,L)$ be a possibilistic Kripke structure, state $s\in S$, $\Phi$, $\Psi$ be PoCTL state formulae, and $\varphi$ be a PoCTL path formula. {\sl The satisfaction relation $\models$} is defined {\sl for state formulae} by
\begin{eqnarray*}
s\models a  & {\rm iff} \ a\in L(s);\\
s\models\neg\Phi & {\rm iff}\ s\not\models\Phi; \\
s\models\Phi\wedge\Psi  & {\rm iff} \ s\models\Phi \ {\rm and}\ s\models\Psi;\\
s\models Po_{J}(\varphi)  & {\rm iff} \ Po(s\models \varphi)\in J, \ {\rm where}\ Po(s\models \varphi)=Po^{M_s}(\{\pi | \pi\in Paths(s), \pi\models \varphi\}).
\end{eqnarray*}

For path $\pi$, {\sl the satisfaction relation $\models$ for path formulae} is defined by
\begin{eqnarray*}
\pi\models\bigcirc\Phi  & {\rm iff} \ \pi[1]\models\Phi;\\
\pi\models\Phi\sqcup\Psi  & {\rm iff} \  \exists k\geq0,\pi[k]\models\Psi
\ {\rm and}\
 \pi[i]\models\Phi {\rm \ for\ all} \ 0\leq i\leq k-1;\\
\pi \models \Phi\sqcup^{\leq n}\Psi & {\rm iff} \ \exists 0\leq k\leq n, (\pi[k]\models \Psi\wedge(\forall 0\leq i< k),\pi[i]\models \Phi)).
\end{eqnarray*}
where if $\pi=s_0s_1s_2\cdots$, then $\pi[k]=s_k$ for any $k\geq 0$.
\end{definition}

In particular, the path formulae $\lozenge\Phi$ (``eventually'') and $\square\Phi$ (``always'') have the semantics
$$\pi=s_0s_1\cdots\models \lozenge\Phi {\rm \ iff}\ s_j\models \Phi {\rm \ for\ some\ } j\geq 0,$$
$$\pi=s_0s_1\cdots\models \square\Phi {\rm \ iff}\ s_j\models \Phi {\rm \ for\ all\ } j\geq 0.$$ Alternatively, $\lozenge\Phi=true\sqcup \Phi$.

\begin{definition}\cite{Xue09} (Syntax of qualitative PoCTL)
{\sl State formulae in the qualitative fragment of PoCTL} (over $AP$) are formed according to
the following grammar:
\begin{center}
$\Phi ::= true\mid a \mid\Phi_{1} \wedge \Phi_{2}\mid \neg \Phi\mid Po_{>0}(\varphi)\mid Po_{=1}(\varphi)$
\end{center}
where $a\in AP$, $\varphi$ is a path formula formed according to the following grammar:
\begin{center}
$\varphi::=\bigcirc \Phi \mid \Phi_{1}\sqcup \Phi_{2}$
\end{center}
where $ \Phi$, $ \Phi_{1}$ and $ \Phi_{2}$ are state formulae.
\end{definition}

As a subclass of PoCTL, the semantics of qualitative PoCTL can be defined as that of PoCTL.

Since we shall compare the expressiveness of PoCTL and CTL, let us recall the definition of CTL.

\begin{definition}\cite{Baier08}(Syntax of CTL)
State formulae in the fragment of CTL (over $AP$) are formed according to
the following grammar:
\begin{center}
$\Phi ::= true\mid a \mid\Phi_{1} \wedge \Phi_{2}\mid \neg \Phi\mid \exists\varphi \mid \forall\varphi$
\end{center}
where $a\in AP$, $\varphi$ is a path formula formed according to the following grammar:
\begin{center}
$\varphi::=\bigcirc \Phi \mid \Phi_{1}\sqcup \Phi_{2}$
\end{center}
where $ \Phi$, $ \Phi_{1}$ and $ \Phi_{2}$ are state formulae.
\end{definition}

\begin{definition}\label{def:semantics of CTL}\cite{Baier08} (Semantics of CTL) Let $a\in AP$ be an atomic proposition, $M=(S,P,I,AP,L)$ be a Kripke structure without terminal state (i.e., $\forall s\in S$, $\exists s^{\prime}\in S$, $(s,s^{\prime})\in P$), state $s\in S$, $\Phi$, $\Psi$ be CTL state formulae, and $\varphi$ be a CTL path formula. The satisfaction relation $\models$ is defined for state formulae by
\begin{eqnarray*}
s\models a  & {\rm iff} \ a\in L(s);\\
s\models\neg\Phi & {\rm iff}\ s\not\models\Phi \\
s\models\Phi\wedge\Psi  & {\rm iff} \ s\models\Phi \ {\rm and}\ s\models\Psi;\\
s\models \exists\varphi  & {\rm iff} \ \pi\models \varphi\ {\rm for}\ {\rm some}\ \pi\in \ Paths(s);\\
s\models \forall\varphi  & {\rm iff} \ \pi\models \varphi\ {\rm for}\ {\rm all}\ \pi\in \ Paths(s).
\end{eqnarray*}

For path $\pi$, the satisfaction relation $\models$ for path formulae is defined by
\begin{eqnarray*}
\pi\models\bigcirc\varphi  & {\rm iff} \ \pi[1]\models\varphi;\\
\pi\models\Phi\sqcup\Psi  & {\rm iff} \  \exists k\geq0,\pi[k]\models\Psi
\ {\rm and}\
 \pi[i]\models\Phi {\rm \ for\ all} \ 0\leq i\leq k-1.
\end{eqnarray*}
\end{definition}

\begin{remark}\label{re:problems} Since we use the PoCTL formula $Po_J(\varphi)$ to denote the possibility measure of the paths satisfying $\varphi$, i.e., $s\models Po_J(\varphi)$ iff $Po(s\models \varphi)\in J$, PoCTL is a possibility measure extension of classical CTL. Both the possibilistic and probabilistic CTL solve certain uncertainty of errors or other stochastic behaviors occurring in various real-world applications. As shown in \cite{Baier08}, probabilistic CTL and CTL are not comparable with respect to their expressiveness. This allows probabilistic CTL to be used to solve the model-checking problems of real-world applications, which can not be tackled by classical model-checking algorithms. With regard to expressiveness of PoCTL, there was no further results on the comparisons between possibilistic CTL and classical CTL. We did not know whether PoCTL can express CTL or vise versa. We shall study the expressiveness of PoCTL in the next section and discuss PoCTL model checking then.
\end{remark}

\section{The expressiveness of PoCTL}

In this section, we study how to define the equivalence between PoCTL formulae and  CTL formulae. We intend to  discuss the equivalence of PoCTL formulae and  CTL formulae and  resolve the problem  whether any  PoCTL formula can be expressed by a  CTL formula or not.

In this section, we always assume that $M$ is a finite possibilistic Kripke structure.

\begin{definition}\label{def:satisfaction set}
 For a possibilistic  Kripke structure $M$ with state space $S$, if $\Phi$ is a state formula, let $Sat_{M}(\Phi)$, or briefly $Sat(\Phi)$, denote $\{s\in S\mid s\models\Phi\}$.
\end{definition}

\begin{definition}\label{def:equivalence of poctl}
 PoCTL formulae $\Phi$ and $\Psi$ are called equivalent, denoted $\Phi\equiv\Psi$, if $Sat(\Phi)=Sat(\Psi)$ for all finite possibilistic Kripke structures $M$ over $AP$.
\end{definition}

\begin{definition}\label{def:equivalence of poctl and ctl}
 A PoCTL formula $\Phi$ is equivalent to a CTL formula $\Psi$, denoted $\Phi\equiv\Psi$, if $Sat_{M}(\Phi)=Sat_{TS(M)}(\Psi)$ for any  finite  possibilistic Kripke structure $M=(S,P,I,AP,L)$, where $TS(M)=(S,\rightarrow,I^{\prime},AP,L)$ is defined by $s\rightarrow s'$ iff $Po(s,s')>0$, and $s\in I^{\prime}$ iff $I(s)>0$. Obviously, $Paths_M(s)=Paths_{TS(M)}(s)$, wo we use the same symbol $Paths(s)$ to denote $Paths_M(s)$ and $Paths_{TS(M)}(s)$ in the following.
\end{definition}

\begin{remark}\label{re:equivalence of poctl and ctl}
Definition \ref{def:equivalence of poctl and ctl} is a key notion, analogous to the one for probabilistic CTL. There are other ways to define an equivalence between CTL and PoCTL formulae. We shall give some discussion of this topic in Section 3.4.
\end{remark}

\begin{theorem}\label{th:negation}
 Let $p\in [0,1]$ be a rational number, $\varphi$ an arbitrary PoCTL path formula, then, we have
\begin{eqnarray}\label{eq:negation}
Po_{<p}(\varphi)\equiv \neg Po_{\geq p}(\varphi).
\end{eqnarray}
\end{theorem}

\begin{proof}
For any $p\in [0,1]$, for any possibilistic Kripke structure $M$ with state space $S$, we have
\begin{eqnarray*}
 Sat(Po_{<p}(\varphi))&=& \{s\mid Po(s\models \varphi)< p\}\\
  &=& S-\{s\mid Po(s\models\varphi)\geq p\}\\
  &=& S-Sat(Po_{\geq p}(\varphi))\\
 &=&Sat(\neg Po_{\geq p}(\varphi)).
\end{eqnarray*}
The last equality follows from the fact $Sat(\neg\Phi)=S-Sat(\Phi)$ for any PoCTL state formula $\Phi$. Therefore, $Po_{<p}(\varphi)\equiv \neg Po_{\geq p}(\varphi)$.
\end{proof}

Dual to Theorem \ref{th:negation}, we have
\begin{eqnarray}\label{eq:dual negation}
Po_{>p}(\varphi)\equiv \neg Po_{\leq p}(\varphi)
\end{eqnarray}
for any rational number $p\in [0,1]$ and path formula $\varphi$. Then it is easy to prove that $$Po_{(p,q)}(\varphi)\equiv\neg Po_{\leq p}(\varphi)\wedge\neg Po_{\geq q}(\varphi).$$
Although the qualitative fragment of PoCTL state formulae only allows possibility bounds of the form $>0$ and $=1$, bounds of the form $=0$ and $<1$ are also definable as $$Po_{=0}(\varphi)\equiv \neg Po_{>0}(\varphi),\ Po_{<1}(\varphi)\equiv\neg Po_{=1}(\varphi).$$

\subsection{CTL formulae are equivalent to PoCTL formulae}

\begin{theorem}\label{th:exists}
Let $\varphi$ be any CTL path formula. Then, we have
\begin{eqnarray}\label{eq:exists}
\exists\varphi\equiv Po_{>0}(\varphi).
\end{eqnarray}
\end{theorem}

\begin{proof}
 Let $M$ be a  finite possibilistic Kripke structure, then we have $Sat_M(Po_{>0}(\varphi))=\{s\mid Po(s\models\varphi)>0\}$, and $Sat_{TS(M)}(\exists\varphi)=\{s\mid\exists\pi\in Paths(s),\pi\models\varphi\}$.

Assume $s\in Sat(Po_{>0}(\varphi))$, then, state $s$ satisfies  $Po(s\models\varphi)>0$, and it follows that $\{s\mid\exists\pi\in Paths(s),\pi\models\varphi\}\not=\emptyset$, i.e.,  $s\in Sat_{TS(M)}(\exists\varphi) $. Therefore, $Sat_M(Po_{>0}(\varphi))\subseteq Sat_{TS(M)}(\exists\varphi)$.

Conversely, if $s\in Sat_{TS(M)}(\exists\varphi)$, then $\exists\pi\in Paths(s),\pi\models\varphi$. Since $M$ is finite and $\pi\in Paths(s)$, it follows that $Po^{M_s}(\pi)>0$, and thus $Po(s\models\varphi)\geq Po^{M_s}(\pi)>0$. Therefore, $s\in Sat_M(Po_{>0}(\varphi))$. This shows that $Sat_{TS(M)}(\exists\varphi)\subseteq Sat_M(Po_{>0}(\varphi))$.

The above shows that $Sat_{TS(M)}(\exists\varphi)=Sat_M(Po_{>0}(\varphi)$. Therefore, we have the required equality.
\end{proof}

To show the further relationship between CTL and PoCTL, we need the existential normal form of CTL formulae.

\begin{definition}\label{def:ENF}\cite{Baier08} For $a\in AP$, the set of CTL state formulae in existential normal form (ENF, in short) is given by
\begin{center}
$\Phi ::= true\mid a \mid\Phi_{1} \wedge \Phi_{2}\mid \neg \Phi\mid\exists\bigcirc\Phi\mid \exists(\Phi_{1}\sqcup\Phi_{2})\mid\exists\square\Phi$.
\end{center}
\end{definition}

\begin{theorem}\label{th:ENF}\cite{Baier08}
For each CTL formulae there exists an equivalent CTL formulae in ENF.
\end{theorem}

\begin{theorem}\label{th:forall}
For any CTL formula, there exists an equivalent qualitative PoCTL formula.
\end{theorem}

\begin{proof}
By Theorem \ref{th:ENF}, each CTL formula can be transformed into an equivalent formula in ENF. Then, by Theorem \ref{th:exists}, each CTL formula in ENF is equivalent to a qualitative PoCTL formula. Combining Theorem \ref{th:ENF} and Theorem \ref{th:exists}, it follows that each CTL formula is equivalent to a qualitative PoCTL formula.
\end{proof}

Theorem \ref{th:forall} shows that CTL is a subclass of PoCTL.
We concretely write some equivalent formulae as follows, most of which do not hold in probabilistic CTL as declared in \cite{Baier08}.

\begin{proposition}\label{pro:exists}
For any CTL formulae $\Phi$ and $\Psi$, we have

(1) $\exists\lozenge \Phi \equiv Po_{>0}(\lozenge \Phi)$,

(2) $\exists\bigcirc \Phi \equiv Po_{>0}(\bigcirc \Phi)$,

(3)  $\exists\square \Phi \equiv Po_{>0}(\square \Phi)$, and

(4) $\exists(\Phi \sqcup \Psi) \equiv Po_{>0}(\Phi\sqcup \Psi)$.
\end{proposition}

\begin{proposition}\label{pro:forall}
For any CTL formulae $\Phi$ and $\Psi$, we have

(1) $\forall\bigcirc \Phi\equiv Po_{=0}(\bigcirc\neg\Phi)$,

(2) $\forall (\Phi\sqcup \Psi)\equiv Po_{=0}(\neg\Psi\sqcup (\neg\Phi\wedge\neg\Psi))\wedge Po_{=0}(\square\neg\Psi)$,

(3)  $\forall\lozenge\Phi\equiv Po_{=0}(\square \neg\Phi)$, and

(4) $\forall\square\Phi\equiv Po_{=0}(\lozenge\neg\Phi)$.
\end{proposition}

\begin{remark}
The above propositions may not hold in infinite possibilistic Kripke structure. We  give a counterexample for  Proposition \ref{pro:forall} (3).

Assume Proposition \ref{pro:forall} (3) holds in any infinite possibilistic  Kripke structure $M$ for $\Phi=a\in AP$, that is $\forall\lozenge a\equiv Po_{=0}(\square \neg a)$ such that state $s$ fulfills both the formula $Po_{=0}(\square \neg a)$ and $\forall\lozenge a$ or none of them.  Fig.1 gives an infinite possibilistic Kripke structure $M=(S,P,I,AP,L)$, in which states are represented by nodes and transitions by labeled edges.
State names are depicted inside the ovals. Initial states are indicated by having an incoming arrow without source. We can see that $Paths(s_{0})=\{s_{0}s_{1}s_{2}\cdots s_{k}t^{w} | k\geq 0\}$. For this $M$, we have $Po(s_{0}\models\square \neg a)=\vee Po\{\pi\in Paths(s_{0})\mid\pi\models\square \neg  a\}=0$, and it follows that $s_{0}\in Sat_M(Po_{=0}(\square \neg  a))$. But $s_{0}s_{1}s_{2}\cdots\not\models\lozenge a$, i.e., $s_{0} \notin Sat_{TS(M)}(\forall\lozenge a)$. This contradicts the assumption that $\forall\lozenge a\equiv Po_{=0}(\square \neg  a)$.

\begin{figure}[ht]
\begin{center}
\includegraphics[scale=0.4]{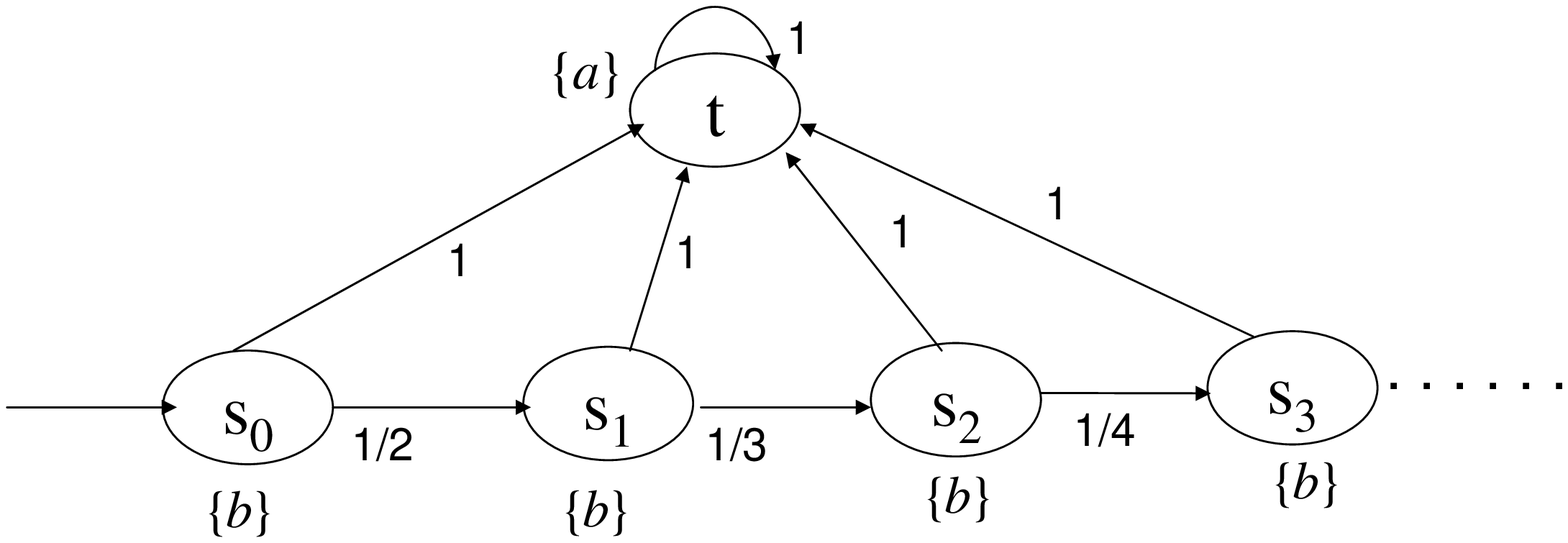}
\center{Fig.1.}An infinite possibilistic Kripke structure $M$.
\vspace{-0.3cm}
\end{center}
\end{figure}
\end{remark}

\subsection{CTL is a proper subclass of PoCTL}

\begin{theorem}\label{th:proper}
There is no CTL formula that is equivalent to $Po_{=1}(\lozenge a)$.
\end{theorem}
\begin{proof}
Assume that there is a CTL  formula $\Phi$ such that $\Phi\equiv Po_{=1}(\lozenge a)$. Consider the following two  finite possibilistic Kripke structures $M_{1}$ and  $M_{2}$, see Fig.2 and Fig.3. By a simple calculation, we have $Po(s_{0}\models\lozenge a)=P(s_{0}s_{1}s_{3}^{w})=1$ in $M_{1}$. However, $Po(s_{0}\models\lozenge a)=Po(s_{0}s_{1}s_{3}^{w})=0.5$ in $M_{2}$. State $s_{0}$ satisfies $Po_{=1}(\lozenge a)$ in $M_{1}$, while $s_{0}$ does not satisfy $Po_{=1}(\lozenge a)$ in $M_{2}$. Hence,
$s_{0}\in Sat_{M_{1}}(Po_{=1}(\lozenge a))$, but $s_{0} \notin Sat_{M_{2}}(Po_{=1}(\lozenge a))$.
This implies that
\begin{eqnarray}\label{eq:not}
Sat_{M_{1}}(Po_{=1}(\lozenge a))\not=Sat_{M_{2}}(Po_{=1}(\lozenge a)).
\end{eqnarray}
Since $\Phi$ is a CTL state formulae, and  $TS(M_{1})=TS(M_{2})$, we have
 \begin{eqnarray}\label{eq:1-2}
Sat_{TS(M_{1})}(\Phi)=Sat_{TS(M_{2})}(\Phi).
\end{eqnarray}
  By the assumption $\Phi\equiv Po_{=1}(\lozenge a)$, it follows that $Sat_{TS(M)}(\Phi)=Sat_{M}(Po_{=1}(\lozenge a))$ for any finite possibilistic Kripke structure $M$. Then we have
 \begin{eqnarray}\label{eq:equal}
Sat_{M_{1}}(Po_{=1}(\lozenge a))=Sat_{M_{2}}(Po_{=1}(\lozenge a)).
\end{eqnarray} Eq.\ref{eq:not} and Eq.\ref{eq:equal} shows a contradiction, which proves that there is no CTL formula that is equivalent to $Po_{=1}(\lozenge a)$.
\end{proof}

\begin{figure}[ht]
\begin{center}
\includegraphics[scale=0.5]{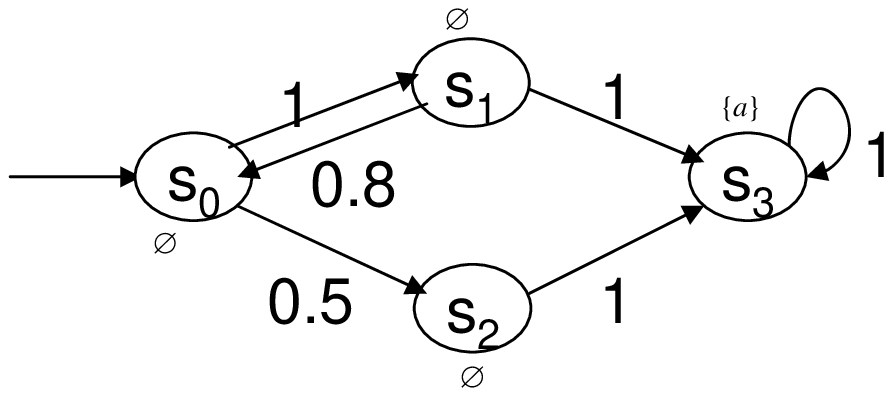}
\center{Fig.2.}A finite possibilistic Kripke structure $M_1$.
\vspace{-0.3cm}
\end{center}
\end{figure}

\begin{figure}[ht]
\begin{center}
\includegraphics[scale=0.5]{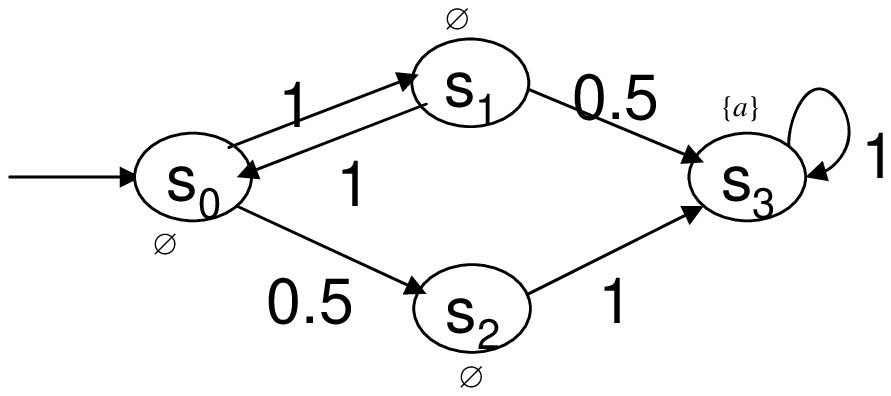}
\center{Fig.3.}A finite possibilistic Kripke structure $M_2$.
\vspace{-0.3cm}
\end{center}
\end{figure}

Combining Theorem \ref{th:forall} and Theorem \ref{th:proper}, it follows that CTL is a proper subclass of PoCTL. PoCTL is completely different from probabilistic CTL. In fact, probabilistic CTL and CTL can not be comparable with each other (whereas, for finite probabilistic Kripke structure, the qualitative fragment of probabilistic CTL can be embedded into CTL and thus a proper subclass of PoCTL).

Using similar arguments, we can show that the following theorems also hold in finite possibilistic Kripke structures.

\begin{theorem}
There is no CTL formula that is equivalent to $Po_{=1}(\bigcirc a)$.
\end{theorem}

\begin{theorem}
There is no CTL formula that is equivalent to $Po_{=1}(\square a)$.
\end{theorem}

\begin{theorem}
There is no CTL formula that is equivalent to $Po_{=1}(a \sqcup b)$.
\end{theorem}

\subsection{ Properties of  repeated reachability and persistence }

 This subsection will show that  qualitative properties for events such as repeated reachability - a certain set of states being visited repeated, and persistence - only a  certain set of states being visited from the moment on, can be described by PoCTL formulae. And we will show that some properties that can not be expressed in CTL can be expressed in the qualitative fragment of PoCTL.

For CTL, universal repeated reachability properties (\cite{Baier08}) can be formalized by the combination of the modalities $\forall \square$ and $\forall\lozenge$:
$$s\models \forall\square\forall\lozenge a\ {\rm iff}\ \pi\models \square\lozenge a \ {\rm for\ all}\ \pi\in Paths(s).$$ For finite possibilistic Kripke structures, a similar result holds for the qualitative fragment of PoCTL.

\begin{theorem}\label{th:repeated reachability}
 Let $M$ be a  finite possibilistic Kripke structure, and  $s$ a state of $M$. Then, we have
\begin{center}
$s\models Po_{=1}(\square Po_{=1}(\lozenge a))$ iff $Po(s\models\square\lozenge a)=1$.
\end{center}
\end{theorem}

\begin{proof}
Since $s\models Po_{=1}(\square Po_{=1}(\lozenge a))$ if and only if $Po(s\models\square Po_{=1}(\lozenge a))=1$,
and $s\models\square Po_{=1}(\lozenge a)$ iff $\pi \models \square Po_{=1}(\lozenge a)$ for any $\pi\in Paths(s)$ ,
 it follows that $Po(s\models\square Po_{=1}(\lozenge a))= Po^{M_s}(\{\pi\in Paths(s)\mid\pi \models\square Po_{=1}(\lozenge a)\})=1$. For any $\pi\models \square Po_{=1}(\lozenge a)$, let $\pi=s_{0}s_{1}\cdots s_{n}\cdots$, then $Po(s_{i}\models \lozenge a)=1$ for any $s_{i}$, where
  $i\geq 0 $. It follows that $\pi\models \square\lozenge a$. Noting that $Po^{M_s}(\pi)\leq Po^{M_s}(\{\pi^{\prime}\in
  Paths(s)\mid\pi^{\prime}\models\square\lozenge a\}$, and thus,
  \begin{center}
   $Po(s\models\square\lozenge a)=Po^{M_s}(\{\pi\in Paths(s)\mid\pi\models\square\lozenge a\})=1$.
   \end{center}

Assume that $Po(s\models\square\lozenge a)=1$. As $Po(s\models\square\lozenge a)=Po^{M_s}(\{\pi\in Paths(s)\mid\pi\models\square\lozenge a\})$ and $M$ is finite, there exists a path $\pi\models\square\lozenge a$ satisfying $Po^{M_s}(\pi)=1$. Let $\pi=s_{0}s_{1}s_{2}\cdots$. Since $\pi\models\square\lozenge a$, we have $\pi[j\cdots]\models\lozenge a$ for any $j\geq 0$, where $\pi[j\cdots]=s_js_{j+1}\cdots$. As $Po^{M_s}(\pi[j\cdots])\geq Po^{M_s}(\pi)$ and $Po^{M_s}(\pi)=1$, it follows that $Po^{M_s}(\pi[j\cdots])=1$ for any $j\geq 0$. Note that $Po^{M_s}(\pi[j\cdots])\leq Po(s_j\models \lozenge a)$, we have $Po(s_j\models \lozenge a)=1$ for any $j\geq 0$. Therefore, we have $Po(s_{0}\models Po_{=1}(\lozenge a))=1$. Hence, $s\models Po_{=1}(\square Po_{=1}(\lozenge a))$.

According to the  above proof, we have:
\begin{center}
$s\models Po_{=1}(\square Po_{=1}(\lozenge a))$ iff $Po(s\models\square\lozenge a)=1$.
\end{center}\end{proof}

In a similar way, by the analysis of the possibility of the  evens such as repeated reachability and persistence with more than $0$ and  equal to $1$, we can show that the following theorems hold in finite possibilistic Kripke structures for atomic events.

\begin{theorem}\label{th:exists repeated reachability}
Let $M$ be a finite possibilistic Kripke structure, and  $s$ a state of $M$. Then, we have
\begin{center}
$s\models Po_{>0}(\square Po_{>0}(\lozenge a))$ iff $Po(s\models\square\lozenge a)>0$.
\end{center}
\end{theorem}

Recall that universal persistence properties can not be expressed in CTL (\cite{Baier08}). For finite possibilistic Kripke structures, PoCTL allows specifying persistence properties with possibility $1$. This is stated by the following theorem.

\begin{theorem}\label{repeated persistence}
Let $M$ be a finite possibilistic Kripke structure,  and $s$ a state of $M$. Then, we have
\begin{center}
$s\models Po_{=1}(\lozenge Po_{=1}(\square a))$ iff $Po(s\models\lozenge\square a)=1$.
\end{center}
\end{theorem}

\begin{theorem}\label{th:exists repeated persistence}
 Let $M$ be a finite possibilistic Kripke structure,  and $s$ a state of $M$. Then, we have
\begin{center}
$s\models Po_{>0}(\lozenge Po_{>0}(\square a))$ iff $Po(s\models\lozenge\square a)>0$.
\end{center}
\end{theorem}

\subsection{Alternative way to define the equivalence between CTL and PoCTL formulae}

As mentioned in Remark \ref{re:equivalence of poctl and ctl}, the definition of the equivalence of PoCTL and CTL formulae is not unique. Definition \ref{def:equivalence of poctl and ctl} is an analogous version of the related definition of probabilistic CTL and CTL formulae. We will give another way to define the equivalence of PoCTL and CTL formulae in the following manner.

\begin{definition}\label{de:alternative equivalence of poctl and ctl}
For a finite possibilistic Kripke structure $M=(S,P,I,AP, L)$ and $\alpha\in (0,1]$, let $TS_{\alpha}(M)=(S,\rw_{\alpha},I_{\alpha},AP,L)$, where $s\rw_{\alpha}t$ iff $P(s,t)\geq \alpha$, and $s\in I_{\alpha}$ iff $I(s)\geq \alpha$. PoCTL formula $\Phi$ is $\alpha$-equivalent to CTL formula $\Psi$, denoted by $\Phi\equiv_{\alpha}\Psi$, if $Sat_M(\Phi)=Sat_{TS_{\alpha}(M)}(\Psi)$ for any finite possibilistic Kripke structure $M$.
\end{definition}

We shall give some properties of PoCTL using the definition of $\alpha$-equivalence of PoCTL and CTL formulae for $\alpha\in (0,1]$. The proofs are very similar to those in Section 3.2.

\begin{proposition}\label{pro:exists}
Let $\varphi$ be any CTL path formula and $\alpha\in (0,1]$. Then, we have
\begin{eqnarray}\label{eq:alternative exists}
\exists\varphi\equiv_{\alpha} Po_{\geq \alpha}(\varphi).
\end{eqnarray}
\end{proposition}

\begin{proposition}\label{pro:alternative equivalence}
For any CTL formula and $\alpha\in (0,1]$, there exists an $\alpha$-equivalent PoCTL formula.
\end{proposition}

\begin{proposition}\label{pro:alternative exists 1}
For any CTL formulae $\Phi$ and $\Psi$, let $\alpha\in (0,1]$, we have

(1) $\exists\lozenge \Phi \equiv_{\alpha} Po_{\geq \alpha}(\lozenge \Phi)$,

(2) $\exists\bigcirc \Phi \equiv_{\alpha} Po_{\geq \alpha}(\bigcirc \Phi)$,

(3)  $\exists\square \Phi \equiv_{\alpha} Po_{\geq \alpha}(\square \Phi)$, and

(4) $\exists(\Phi \sqcup \Psi) \equiv_{\alpha} Po_{\geq \alpha}(\Phi\sqcup \Psi)$.
\end{proposition}

\begin{proposition}\label{pro:alternative forall}
For any CTL formulae $\Phi$ and $\Psi$, let $\alpha\in (0,1]$, we have

(1) $\forall\bigcirc \Phi\equiv_{\alpha} Po_{<\alpha}(\bigcirc\neg\Phi)$,

(2) $\forall (\Phi\sqcup \Psi)\equiv_{\alpha} Po_{<\alpha}(\neg\Psi\sqcup (\neg\Phi\wedge\neg\Psi))\wedge Po_{<\alpha}(\square\neg\Psi)$,

(3)  $\forall\lozenge\Phi\equiv_{\alpha} Po_{<\alpha}(\square \neg\Phi)$, and

(4) $\forall\square\Phi\equiv_{\alpha} Po_{<\alpha}(\lozenge\neg\Phi)$.
\end{proposition}

\begin{proposition}\label{pro:alternative proper}
For any $\alpha\in (0,1]$, there is no CTL formula that is ${\alpha}$-equivalent to $Po_{=1}(\lozenge a)$.
\end{proposition}

The $\alpha$-equivalence of PoCTL and CTL formulae might be useful in the approximation of PoCTL formulae using CTL formulae. This would allow a graded approach to establish a level cut to decide e.g. when a transition with value $\alpha$ can be considered as existing or not. The general notion of $\alpha$-equivalence would be a very general approach such that the notions of equivalence (actually $>0-$equivalence) and 1-equivalence would come out as a limit case and particular case respectively. However, intuitively, 1-equivalence is too strong to define the equivalence of PoCTL and CTL formulae in the senses as explained below. By 1-equivalence, the possibility of a certain ``event'' is larger than 0 does not imply that the  ``event'' exists. For example, in Fig.3, intuitively, $s_0\models \exists\lozenge a$. However, by a simple calculation , we have $Po(s_0\models \lozenge a)=0.5<1$. It follows that $s_0\not\models Po_{=1}(\lozenge a)$, hence,  $s_0\not\models \exists\lozenge a$. Furthermore, intuitively, 1-equivalence is too strong for universal quantifier $\forall$. By Proposition \ref{pro:alternative forall}, the universal ``event'' means that the possibility of the negation of the ``event'' is less than 1. There are ``events'' such that the possibility of the negation of the ``events'' is less than 1 but there exist some paths that violate the ``events''. We shall give some analysis in the illustrative example in Section 5.

\section{PoCTL Model Checking}

Similar to classical and probabilistic CTL model-checking problems, the PoCTL model-checking problem can be stated as follows:

For a given finite possibilistic
Kripke structure $M$, state $s$ in $M$, and PoCTL state formula $\Phi$, decide whether $s\models\Phi$.

 We write $(M,s)\models \Phi$ for this PoCTL model-checking problem.

 As shown in the above section, PoCTL is more expressible than CTL. There are some PoCTL model-checking problems that can not be tackled by classical CTL model-checking algorithm. We shall present some methods to tackle PoCTL model-checking problems in this section. The technique of PoCTL model checking is very similar to those of classical and probabilistic CTL model checking. The difference lies in the operations involving in the processing of model checking.

 To determine whether $s\models\Phi$, we need to compute the satisfaction set $Sat(\Phi)$.
This is done recursively using a bottom-up traversal of the parse tree of  $\Phi$ with time complexity ${\cal O}(|\Phi|)$, where $|\Phi|$ denotes the number of subformulae of $\Phi$ (see the definition of $|\Phi|$ in Section 6.4.3 in \cite{Baier08}).  As for CTL model checking, the nodes of the parse tree represent the subformulae of $\Phi$. For each node
of the parse tree, which represents a subformula  $\Psi$ of  $\Phi$, the set  $Sat(\Psi)$ is calculated. If $\Psi$ is propositional logic formula, $Sat(\Psi)$ can be computed in exactly the same way as for
CTL. The left part is the treatment of subformulae of the form $\Psi=Po_{J}(\varphi)$. Since
\begin{eqnarray}\label{eq:sat-po}
Sat(Po_{J}(\varphi))=\{s\in S\mid Po(s\models\varphi)\in J\},
\end{eqnarray}
to calculate $Sat(\Psi)$, we need to compute the possibility $Po(s\models\varphi)$ for any state $s$.

There are three ways  to construct path formula $\varphi$, i.e., $\varphi=\bigcirc\Psi$, $\varphi=\Phi\sqcup^{\leq n}\Psi$ or $\varphi=\Phi\sqcup\Psi$ for some state formulae $\Phi$ and $\Psi$ and $n\in \mathbb{N}.$

For $\varphi=\bigcirc\Psi$, the next-step operator, the following equality holds:
\begin{center}
$Po(s\models\bigcirc\Psi)=\bigvee\limits_{s'\in Sat(\Psi)}P(s,s')$
\end{center}
where $P$ is the transition matrix of $M$. In the matrix-vector notation we thus
have that the (column) vector $(Po(s\models\bigcirc\Psi))_{s\in S}$ can be computed by multiplying $P$ with the characteristic
vector for $Sat(\Psi)$, i.e., (column) bit vector $(b_{s})_{s\in S}$ where $b_{s} = 1$ if and only if $s\in Sat(\Psi)$. Write $\chi_{\Psi}=(b_{s})_{s\in S}$, then we have
\begin{eqnarray}\label{eq:next operator}
(Po(s\models\bigcirc\Psi))_{s\in S}=P\circ \chi_{\Psi}.
\end{eqnarray}
It follows that, checking the next-step operator thus reduces to a single matrix-vector multiplication.

To calculate the possibility $Po(s\models\varphi)$ for until formulae $\varphi=\Phi\sqcup^{\leq n}\Psi$ or $\varphi=\Phi\sqcup\Psi$. Let $C=Sat(\Phi)$ and $B=Sat(\Psi)$, by its definition, we have
$$Po(s\models \Phi\sqcup^{\leq n}\Psi)=Po(s\models C\sqcup^{\leq n} B), \ {\rm and}$$
$$Po(s\models \Phi\sqcup\Psi)=Po(s\models C\sqcup B),$$
where $Po(s\models C\sqcup^{\leq n} B)=Po^{M_s}(\{\pi\in Paths(s) | \exists 0\leq j\leq n, \pi[j]\in B$ and for any $0\leq k<j$, $\pi(k)\in C\})$ and $Po(s\models C\sqcup B)=Po^{M_s}(\{\pi\in Paths(s) | \exists j\geq 0, \pi[j]\in B$ and for any $0\leq k<j$, $\pi(k)\in C\})$

We posed a least fixed point characterization to calculate $Po(s\models C\sqcup B)$ in \cite{li12}. In the following, we shall give a direct method to calculate $Po(s\models C\sqcup^{\leq n} B)$ and $Po(s\models C\sqcup B)$, which is completely different from the method used in probabilistic CTL model checking for until operator, where a linear equation system needs to be solved with more time complexity.

As done in \cite{li12}, let $S_{=0}, S_{=1}, S_{?}$ be a partition of $S$ such that,

 (1) $B\subseteq S_{=1} \subseteq\{s\in S| Po(s\models C\sqcup B)=1\}$;

 (2) $S\backslash (C\cup B)\subseteq S_{=0} \subseteq\{s\in S| Po(s\models C\sqcup B)=0\}$;

(3) $S_{?}=S\backslash(S_{=1}\cup S_{=0})$.

The above partition of $S$ always exists. For example, we can take $S_{=1}=B$, $S_{=0}=S\backslash (C\cup B)$ and $S_{?}=S\backslash(S_{=1}\cup S_{=0})=C-B$. Note that the technique and notations used here have been adopted from probabilistic CTL model checking \cite{Baier08}.

      For all state $s$, write $$x_s=Po(s\models C\sqcup^{\leq n}B).$$ If $s\in S_{=1}$, we have $Po(s\models C\sqcup^{\leq n} B)=1$; if $s\in S_{=0}$, $Po(s\models C\sqcup^{\leq n} B)=0$; if $s\in S_{?}$, we can get a
fuzzy matrix $P_?=(P_?(s,t))_{s,t\in S}$ by letting $P_?(s,t)=P(s,t)$ whenever $s,t\in S_?$ and $0$ otherwise. The left is to give a method to calculate $(x_s)_{s\in S_?}$.

By the definition of $C\sqcup^{\leq n}B$, we have
\begin{eqnarray*}
&&\{\pi\in Paths(s) | \pi\models C\sqcup^{\leq n} B\}\\
&=&\{\pi\in Paths(s) | \exists k\leq n,\ {\rm if}\ 0\leq i<k, \pi(i)\in C, \ {\rm and}\ \pi(k)\in B\}\\
&=& \bigcup\{Cyl(s_0\cdots s_kt) | 0\leq k\leq n, s_0=s, s_1,\cdots,s_k\in C\ {\rm and}\ t\in B\}.
 \end{eqnarray*}
Hence, \begin{eqnarray*}
&&Po(s\models C\sqcup^{\leq n} B)\\
&=&\bigvee_{k=0}^n\bigvee\{Po(Cyl(s_0\cdots s_k t) | s_0=s, s_1,\cdots,s_k\in C\ {\rm and}\ t\in B\}.
 \end{eqnarray*}

Write $\chi_s=(a_t)_{t\in S_?}$ for the (row) characteristic vector for the singleton $\{s\}$, i.e., $a_t=1$ if $t=s$ and $a_t=0$ if $t\not=s$; $\chi_B=(b_t)_{t\in S_?}$ for the (column) characteristic vector for $B$, i.e., $b_t=1$ if $t\in B$ and $0$ otherwise. By a simple calculation, we have
$$\bigvee\{Po(Cyl(s_0\cdots s_k t) |s_0=s, s_1,\cdots,s_k\in C\ {\rm and}\ t\in B\}=\chi_s\circ P_?^{k}\circ P \circ \chi_B$$ for any $k$. It follows that $$x_s=Po(s\models C\sqcup^{\leq n} B)=\bigvee_{k=0}^n\chi_s\circ P_?^k\circ P \circ \chi_B=\chi_s\circ \bigvee_{k=0}^n P_?^k\circ P \circ \chi_B.$$ If we write $P_?^{\leq n}=\bigvee_{k=0}^n P_?^k$, where $P_?^0$ is the identity matrix, i.e., $P_?^0(s,s)=1$ and $0$ otherwise, then $$x_s=Po(s\models C\sqcup^{\leq n} B)=\chi_s\circ  P_?^{\leq n}\circ P \circ \chi_B.$$

Hence, if we write $\chi_?=(\chi_?(s,t))_{s\in {S_?},t\in S}$ as the characteristic matrix for $S_?$ in $S$, i.e., $\chi_?(s,s)=1$ for $s\in S_?$ and $0$ otherwise, then we have
\begin{eqnarray}\label{eq:until-n?}
(x_s)_{s\in S_?}=\chi_?\circ P_?^{\leq n}\circ P\circ \chi_B.
\end{eqnarray}
To calculate $(x_s)_{s\in S_?}$, it is sufficient to perform matrix multiplication at most $n+3$ times. Observe that, if $n\geq |S_?|$, then $P_?^{\leq n}=P_?^0\vee P_?^+$, which is denoted by $P_?^{\ast}$. Then $P_?^{\ast}$ is the reflexive and transitive closure of the fuzzy matrix $P_?$. In this case, we have
\begin{eqnarray}\label{eq:until-n-?}
(x_s)_{s\in S_?}=\chi_?\circ P_?^{\ast}\circ P\circ \chi_B.
\end{eqnarray} In particular, we have
\begin{eqnarray}\label{eq:until-?}
(x_s)_{s\in S_?}=(Po(s\models C\sqcup B))_{s\in S_?}=\chi_?\circ P_?^{\ast}\circ P\circ  \chi_B.
\end{eqnarray}

In summary, we have
\begin{eqnarray}\label{eq:until-n}
 x_s=Po(s\models C\sqcup^{\leq n} B)=\left\{
\begin{array}{cc}
1,& $if$\ s\in S_{=1},\\
0,& $if$\ s\in S_{=0},\\
\chi_s\circ P_?^{\leq n}\circ P \circ \chi_B,& $if$\ s\in S_?.\\
\end{array}
\right.
 \end{eqnarray}

 In particular, if $n\geq |S_?|$, we have
 \begin{eqnarray}\label{eq:until}
 x_s=Po(s\models C\sqcup^{\leq n} B)=Po(s\models C\sqcup B)=\left\{
\begin{array}{cc}
1,& $if$\ s\in S_{=1},\\
0,& $if$\ s\in S_{=0},\\
\chi_s\circ P_?^{\ast} \circ P\circ \chi_B,& $if$\ s\in S_?.\\
\end{array}
\right.
 \end{eqnarray}

In the calculation of $(x_s)_{s\in S}$, we only need to perform (fuzzy) matrix multiplication at most $N(=|S|)+3$ times. It follows that the time complexity of PoCTL model checking of a finite possibilistic Kripke structure $M$ and a PoCTL formula $\Phi$ can be presented as follows.

\begin{theorem}\label{th:time of PoCTL} ({\rm{Time Complexity of PoCTL Model Checking}})
For a finite possibilistic Kripke structure $M$, state $s$ in $M$, and a PoCTL formula $\Phi$, the PoCTL model-checking problem $(M,s)\models \Phi$ can be determined in time ${\cal O}(size(M)\cdot N\cdot |\Phi|)$, where $|\Phi|$ denotes the number of subformulae of $\Phi$.
\end{theorem}

\section{An illustrative example}

We now give an example to illustrate the PoCTL model-checking approach presented in this paper. The same example is used in \cite{li12} to illustrate the application of model checking of linear-time properties based on possibility measures. Note that this is a demonstrative rather than a case study aimed at showing the scalability of our approach.

Suppose that there is an animal with a new disease. For the new disease, the doctor has no complete knowledge about it, but he (or she) believes by experience that the drug Ribavirin may be useful for the treating the disease.

For simplicity, it is assumed that the doctor considers roughly the animal's condition to be three states, say, ``poor'', ``fair'' and ``excellent''. It is vague when the animal's condition is said to be ``poor'', ``fair'' and ``excellent''. The doctor will use the fuzzy set (called fuzzy state in the following) over states ``poor'', ``fair'' and ``excellent''   to describe the animal condition (see \cite{lin02,cao06,liu09} for more explanations). Therefore, when a possibilistic Kripke structure is used to model the treatment processes of the animal, a fuzzy state is naturally denoted as a three-dimensional vector $[a_1, a_2, a_3]$, which is represented as the possibility distribution of the animal's condition over states ``poor'', ``fair'' and ``excellent''.

Similarly, it is imprecise to say that at what point exactly the animal has changed from one state to another state after a drug treatment (i.e., event), because the drug event occurring may lead a state to fuzzy state ``poor'', ``fair'' and ``excellent''. Therefore, the treatment process is modeled by a possibilistic Kripke structure, in which a transition possibility distribution is represented by a $3\times 3$ matrix.

Suppose that the treatment process of the animal is modeled by the following possibilistic Kripke structure $M=(S,P,I,AP,L)$, where
$S=AP=\{poor, fair$, $excellent\}$,

$P=\left(\begin{array}{cccc}
0.2&1&1\\
0.2&0.5&1\\
0.5&1&0.5
\end{array}
\right)$,
$I=\left(\begin{array}{cccc}
1\\
0\\
0
\end{array}
\right)$,

\noindent and $L(s)=\{s\}$ for any $s\in S$.

The structure $M$ is presented in Fig.4, and the corresponding $M^+$ is presented in Fig.5, where we use the symbols $p, f, e$ to represent the states or the atomic propositions ``poor'', ``fair'' and ``excellent'' respectively.

\begin{figure}[ht]
\begin{center}
\includegraphics[scale=0.5]{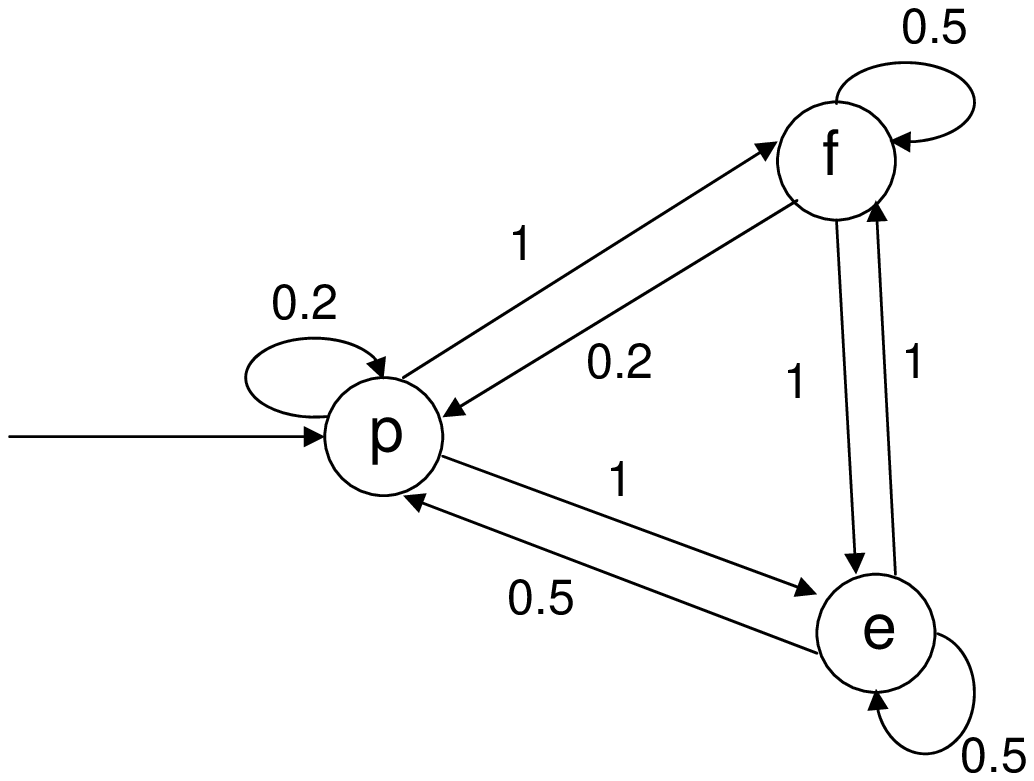}
\center{Fig.4.}The possibilistic Kripke structure $M$ for the treatment process of the animal.
\vspace{-0.3cm}
\end{center}
\end{figure}

\begin{figure}[ht]
\begin{center}
\includegraphics[scale=0.5]{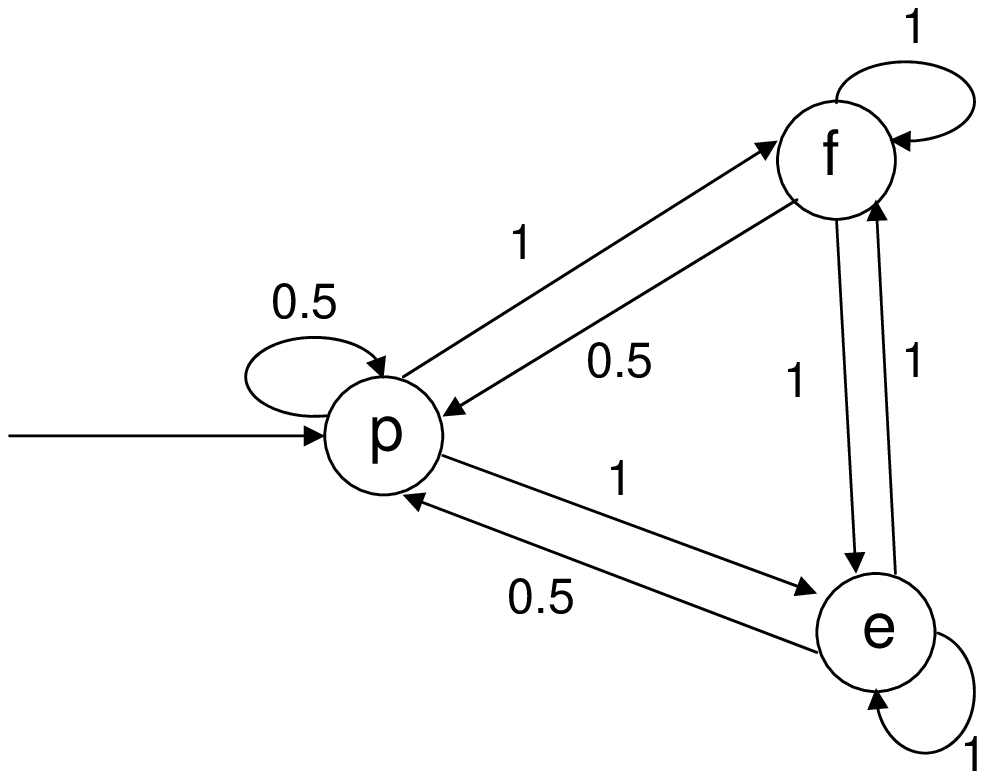}
\center{Fig.5.}The corresponding $M^+$ of $M$ in Fig.4.
\vspace{-0.3cm}
\end{center}
\end{figure}

By a simple calculation, we have

$P^+=\left(\begin{array}{cccc}
0.5&1&1\\
0.5&1&1\\
0.5&1&1
\end{array}
\right)$.

Some calculations are presented as follows in detail.

(1) Let us calculate $Po(poor\models \{poor\}\sqcup^{\leq 7}\{excellent\})$. In this case, let us take $S_{=1}=\{excellent\}$, $S_{=0}=\{fair\}$, and $S_{?}=\{poor\}$. It follows that,
$P_?=\left(\begin{array}{cccc}
0.2&0&0\\
0&0&0\\
0&0&0
\end{array}
\right)$, and then
$P_?^{\ast}=\left(\begin{array}{cccc}
1&0&0\\
0&1&0\\
0&0&1
\end{array}
\right)$.
By Eq.\ref{eq:until}, we have
\begin{eqnarray*}
&& Po(poor\models \{poor\}\sqcup^{\leq 7}\{excellent\})\\
&=& \left(\begin{array}{ccc}
                         1&0&0
                       \end{array}
                       \right)\circ P_?^{\leq 7}\circ P \circ \left(\begin{array}{c}
                         0 \\
                         0 \\
                         1
                       \end{array}
                       \right)\\
&=& \left(\begin{array}{ccc}
                         1&0&0
                       \end{array}
                       \right)\circ P_?^{\ast} \circ P \circ \left(\begin{array}{c}
                         0 \\
                         0 \\
                         1
                       \end{array}
                       \right)\\
&=&1.\\
 \end{eqnarray*}
  Hence, $poor\models Po_{=1}(\{poor\}\sqcup^{\leq 7}\{excellent\})$. It means that the animal will be recovered after one week treatment with possibility $1$.

(2) Since \begin{eqnarray*}
 Po(poor\models \lozenge\{excellent\})&=& Po(poor\models true\sqcup\{excellent\})\\
&=& \left(\begin{array}{ccc}
                         1&0&0
                       \end{array}
                       \right)\circ \left(\begin{array}{cccc}
                         1&0.5&0\\
                         0.2&1&0\\
                         0&0&1\\
                       \end{array}
                       \right)\circ \left(\begin{array}{cccc}
0.2&1&1\\
0.2&0.5&1\\
0.5&1&0.5
\end{array}
\right)\circ \left(\begin{array}{c}
                         0 \\
                         0 \\
                         1
                       \end{array}
                       \right)\\
&=&1.\\
 \end{eqnarray*}
 In this case, we take $S_{=1}=\{excellent\}$, $S_{=0}=\emptyset$ and $S_{?}=\{poor, fair\}$.

 Hence, $poor\models Po_{=1}(\lozenge\{excellent\})$.

 (3) We have $poor\not\models\forall\lozenge\{excellent\}$. The reason is as follows. By Proposition \ref{pro:forall}(3), we have $$\forall\lozenge\{excellent\}\equiv Po_{=0}(\square\neg\{excellent\}).$$ Let us calculate $Po(s\models\square\neg \{excellent\})$, where $s=poor$: $$Po(s\models\square\neg \{excellent\})=Po^{M_s}(\{\pi\in paths(s) | \pi\models \square\neg \{excellent\})=Po^{M_s}(pf^{\omega})=0.5>0.$$ Hence, $s=poor\not\models Po_{=0}(\square\neg\{excellent\})$, i.e., $poor\not\models\forall\lozenge\{excellent\}$.

 Since $\forall\lozenge\{excellent\}\equiv_1 Po_{<1}(\square\neg\{excellent\})$, and $Po(poor\models \square\neg\{excellent\})=0.5<1$, it follows that $poor\models_1 \forall\lozenge\{excellent\}$ if we adopt 1-equivalence. This is too strong, since we still have the event $p^{\omega}$ (with possibility 0.2) and the event $pf^{\omega}$ (with possibility 0.5), and the above two events (may occur) violate the property $\forall\lozenge\{excellent\}$.

 (2) and (3) show that $\forall\lozenge \Phi\equiv Po_{=1}(\lozenge \Phi)$ does not hold in PoCTL.

 (4) Let $s=poor$, $a=excellent$, by Theorem \ref{th:repeated reachability}, we have $$s\models Po_{=1}(\square Po_{=1}(\lozenge a))\ {\rm iff}\ Po(s\models \square\lozenge a)=1.$$ It has been shown that $Po(s\models \square\lozenge a)=1$ in \cite{li12}. Then we know that $$s\models Po_{=1}(\square Po_{=1}(\lozenge a)).$$

 (5) Since $Po(s\models \square\neg \{poor\})=0$, where $s=poor$. It follows that $poor\models \forall\lozenge \{poor\}$.

\section{Conclusion}

This paper is a continuation of previous work in the papers \cite{li12,Xue09}, where LTL model checking based on possibility measures and possibilistic CTL were introduced. We further studied the expressiveness of PoCTL and PoCTL model cheking in this paper, which was not considered in \cite{li12,Xue09}. The main contribution of this paper is as follows. We showed that (qualitative) PoCTL is more powerful than CTL with respect to their expressiveness. In particular, we have shown that any CTL formula is equivalent to a qualitative PoCTL formula. Some basic PoCTL formulae that are not equivalent to any CTL formulae were also given. Some qualitative repeated reachability and persistence properties were expressed using PoCTL formulae. The PoCTL model checking problem was discussed in detail. The method of PoCTL model checking were given and its time complexity was analyzed.

This is the first step of PoCTL model checking. There are many things that can be done based on this.

As we know, there are many industrial model checkers related to CTL model checking, including SMV (\cite{McMillan93}) and NuSMV. Since CTL is a proper subclass of PoCTL, it is necessary to set up some model checker corresponding to PoCTL model checking. The equivalence and abstraction technique corresponding to PoCTL model checking are also necessary to be investigated in the future work.

Of course, the research directions related to possibilistic LTL model checking posed in \cite{li12} can also be
applied to PoCTL model checking. We list three of them as follows.

\begin{itemize}
\item We use max-min composition of fuzzy relations in this paper. There are other forms of composition of fuzzy relations, such as max-product composition which are useful for the applications of fuzzy sets. Then the related work using other composition instead of max-min composition can be done in the future.

\item We use the normal possibility distributions in this paper (see conditions (2) and (3) in the definition of possibilistic Kripke structure). How to deal with those possibilistic Kripke structures which do not satisfy conditions (2) and (3) is another future direction to study.

\item In the definition of possibilistic Kripke structures, the labeling function $L: S\rw 2^{AP}$ is crisp, there is no vagueness at all here. This restriction is too strict. How to dealt with the possibilistic Kripke structures with uncertainty labeling function in PoCTL is still another issue needed to be discussed further. Although we can transform a possibilistic Kripke structure with uncertainty labeling function into a possibilistic Kripke structure with classical labeling function as noted in \cite{li12}, a direct method using possibilistic Kripke structures with uncertainty labeling functions still deserves study.

\end{itemize}

\section*{Acknowledgments}

The authors would like to thank the anonymous referees for helping them refine the ideas
presented in this paper and improve the clarity of the presentation. The authors would also like to express their special thanks to Dr. Licong Cui at Case Western Reserve University for detailed
suggestions that improved the paper's quality.


\begin{thebibliography}{99}
\baselineskip 10pt

\bibitem{Baier08} C. Baier, J. P. Katoen, \emph{Principles of Model Checking}, Cambridge: The MIT Press, 2008.

\bibitem{Baltazar12}P. Baltazar, P. Mateus, R. Nagarajan, N. Papanikolaou,  \emph{Exogenous probabilistic computation tree logic}, Electronic Notes in Theoretical Computer Science, 190(2007) 95-110.

\bibitem{Barbuti00}R. Barbuti, F. Levi, P. Milazzo, G. Scatena, \emph{Probabilistic model checking of biological systems with uncertain kinetic rates},Theoretical Computer Science, 419(2012) 2-16.

\bibitem{cao06} Y.Cao, M. Ying, \emph{Observability and decentralized control of fuzzy discrete-event systems}, IEEE Transactions on Fuzzy Systems,14(2)(2006) 202-216.


\bibitem {Ciesinski18}F.Ciesinski, M.Grober, \emph{On probabilistic computation tree logic}, Lecture Notes in Computer Science, 2925(2004) 333-355.

\bibitem {Clarke08}E.Clarke, E. Emerson, A.Sistla, \emph{Automatic verification of finite-State
concurrent systems using temporal logic
specifications}, ACM Transactions on Programming Languages and Systems, 8(2)(1986) 244-263.

\bibitem{cui08} L. Cui, Y. Li, \emph{Linguistic quantifiers based on Choquet integrals}, International Journal of Approximate Reasoning, 48(2008) 559-582.

\bibitem{Dovier02} A. Dovier, E. Quintarelli, \emph{Applying model-checking to solve queries on semistructured data}, Computer Languages, Systems and Structures, 35(2009) 143-172.

\bibitem{Drakopoulos11} A. Drakopoulos, \emph{Probabilities, possibilities, and fuzzy sets}, Fuzzy Sets and Systems,75(1995) 1-15.

\bibitem{dubois94} D. Dubois , F. Dupin de Saint Cyr, H. Prade, \emph{Updating, transition constraints
and possibilistic Markov chains}, Proc. of IPMU 1994, 826-831, 1994.

\bibitem{Didier14}D. Dubois, H. Prade, \emph{Possibility theory, probability theory and multiple-valued
logics: A clarification}, Annals of Mathematics and Artificial Intelligence, 32(2001) 35-66.

\bibitem{EGP92} M. Edmund, O. Grumberg, D. Peled, \emph{ Model Checking}, Cambridge: the MIT Press, 1999.

\bibitem{grabisch00} M. Grabisch, T. Murofushi, M. Sugeno (eds), \emph{Fuzzy Measures and Integrals}, Heidelberg New Tork: Physica-Verlag, 2000.


\bibitem{hajek95}P. H\'{a}jek, D. Harmancov\'{a}, R. Verbrugge, \emph{A qualitative fuzzy
possibilistic logic}, International Journal of Approximate Reasoning, 12(1995) 1-19.

\bibitem {Hart06}	S. Hart, M. Sharir, \emph{Probabilistic propositional temporal logics}, Information and Control,70(2-3)(1986) 97-155.

\bibitem{janssen96} H. Janssen, G. de Cooman and E.E. Kerre,  \emph{First results for a mathematical theory of possibilistic Markov processes}, in: Proceedings of IPMU'96, Vol. III (Information Processing and
Management of Uncertainty in Knowledge-Based Systems), Granada, Spain (1996) pp. 1425-
1431.

\bibitem{kruse87} R. Kruse, R. Buck-Emden, R. Cordes, \emph{Processor power considerations - An
application of fuzzy Markov chains}, Fuzzy Sets and Systems, 21(1987) 289-299.

\bibitem{li05}Y. Li, \emph{Analysis of Fuzzy Systems(in Chinese)},  Beijing, China: Science Press, 2005.

\bibitem{li12}Y. Li, L. Li, \emph{Model checking of linear-time properties based on possibility measure}, IEEE Transactions on Fuzzy Systems, 21(5)(2013), 842-854.

\bibitem{lin02} F. Lin, H. Ying, \emph{Modeling and control of fuzzy discrete event systems}, IEEE Transactions on Systems, Man, and Cybernetics, Part B, 32(4)(2002) 408-415.

\bibitem{liu09} F. Liu, D. Qiu, \emph{Diagnosability of fuzzy discrete-event systems: a fuzzy approach}, IEEE Transactions on Fuzzy Systems, 17(2)(2009) 372-384.

\bibitem{Patthak10}A.C.Patthak, I.Bhattacharya, A.Dasgupta, Pallab Dasgupta, P.P.Chakrabarti, \emph{Quantified computation tree logic},
 Information Processing Letters, 8(2002) 123-129.

\bibitem{Sugeno74}M.  Sugeno, \emph{Theory of Fuzzy Integrals and its Applications},  \hskip 1em plus 0.5em minus 0.4em\relax PhD thesis, Tokyo Institute of Technology, 1974.

 \bibitem{Li16}K.Y.Rozier, \emph{Linear temporal logic symbolic model checking}, Computer Science Review, 5(2011) 163-203.

 \bibitem{McMillan93} K. McMillan, \emph{Symbolic Model Checking},  Dordrecht, The Netherland: Kluwer, 1993.

\bibitem{Xue09} Y. Xue, H. Lei, Y. Li, \emph{Computationg tree logic based on possibility measure(in Chinese)}, Computer Engineering and Science,33(9)(2011) 70-75.

 \bibitem{ying06}M. Ying, \emph{Linguistic quantifiers modeled by Sugeno integrals}, Artificial Intelligence, 170(6-7)(2006) 581-606.

\bibitem{Zadeh65} L.A. Zadeh, \emph{ Fuzzy sets},  Information and Control, 8(1965) 338-353.
 \bibitem{Zadeh78} L.A. Zadeh, \emph{ Fuzzy sets as a basis for a theory of possibility}, Fuzzy Sets and  Systems,
 1 (1978) 3-28.

\end{thebibliography}
\end{document}